\documentclass[conference]{IEEEtran}
\usepackage[pdftex]{graphicx}
\usepackage{texdraw,epsfig,bm}
\usepackage{amsfonts}
\usepackage{amsmath}
\usepackage{amssymb}
\usepackage{amsthm}
\usepackage{xcolor}
\allowdisplaybreaks 

\newtheorem{thm}{Theorem}

\newtheorem{lem}{Lemma}

\newtheorem{prb}{Problem}

\newcommand{\JG}[1]{{\color{black} #1}}
\newcommand{\BS}[1]{{\color{black} #1}}
\newcommand{\BSwcnc}[1]{{\color{black} #1}}
\newcommand{\BSwcncfin}[1]{{\color{black} #1}}
\newcommand{\BSr}[1]{{\color{black} #1}}
\newcommand{\JGr}[1]{{\color{black} #1}}

\begin{document}
\title{On Optimal Geographical Caching in Heterogeneous Cellular Networks}
\author{\IEEEauthorblockN{Berksan Serbetci}
\IEEEauthorblockA{Stochastic Operations Research\\University of Twente, The Netherlands\\
b.serbetci@utwente.nl}
\and
\IEEEauthorblockN{Jasper Goseling}
\IEEEauthorblockA{Stochastic Operations Research\\University of Twente, The Netherlands\\
j.goseling@utwente.nl}}
\maketitle
\begin{abstract}
In this work we investigate optimal geographical caching in heterogeneous cellular networks \BSr{where} different types of base stations (BSs) have different cache capacities. \BSr{Users request files from a content library according to a known probability distribution.} The performance metric is the total hit probability, which is the probability that a user at an arbitrary location in the plane will find the content that it requires in one of the BSs that it is covered by. 

We consider the problem of optimally placing content in all BSs jointly. As this problem is not convex, we provide a heuristic scheme by finding the optimal placement \BSr{policy for one type of base station conditioned on the placement in all other types}. We demonstrate that these individual optimization problems are convex and we provide an analytical solution. As an illustration, we find the optimal placement policy of the small base stations (SBSs) depending on the placement policy of the macro base stations (MBSs). We show how the hit probability evolves as the deployment density of the SBSs varies. We show that the heuristic of placing the most popular content in the MBSs is almost optimal after deploying the SBSs with optimal placement policies. Also, for the SBSs no such heuristic can be used; the optimal placement is significantly better than storing the most popular content. Finally, we show that solving the individual problems to find the optimal placement policies for different types of BSs iteratively, namely repeatedly updating the placement policies, does not improve the performance.
\end{abstract}
\section{Introduction}
\BSr{Recent years have seen} extreme growth in data traffic over cellular networks. \BSr{The demand} is expected to \JGr{further} increase in the upcoming years such that current network \BS{infrastructures \cite{cisco}} will not be able to support this data traffic \cite{femtocells}. \JGr{One of the bottlenecks is formed by the backhaul and a promising means of reducing}
backhaul traffic \JGr{is} \BS{by} reserving some storage capacity at \BSr{base stations (BSs)} and use these as caches \BSwcnc{\cite{femtod2d}}. In this way, part of the data is stored at the wireless edge and the backhaul is used only to refresh this stored data. Data replacement will depend on the users' demand distribution over time. \BSwcncfin{Since this distribution is varying slowly, the stored data can be refreshed at off-peak times.} \BSr{Thus}, \BSwcnc{caches} containing popular content \BSr{can serve users without incurring a load on the backhaul}.

\BSwcnc{Recently, there has been growing interest in caching in cellular networks.}
In \cite{femtocaching} Shanmugam \emph{et al.} focus on the content placement problem and analyze which files should be cached by which helpers for the given network topology and file popularity distribution by minimizing the expected total file delay.
In \cite{approximation} Poularakis \emph{et al.} provide an approximation algorithm for the problem of minimizing the user content requests routed to macrocell base stations with constrained cache storage and bandwidth capacities.
In \cite{optimalgeographic} B{\l}aszczyszyn \emph{et al.} revisit the optimal content placement in cellular caches by assuming a known distribution of the coverage number and provide the optimal probabilistic placement policy which guarantees maximal total hit probability.
In \cite{fundamental} Maddah-Ali \emph{et al.} developed an information-theoretic lower bound for the caching system for local and global caching gains. In \cite{distrcach} Ioannidis \emph{et al.} propose a novel mechanism for determining the caching policy of each mobile user that maximize the system's social welfare.
In \cite{exploiting} Poularakis \emph{et al.} consider the content storage problem of encoded versions of the content files. In \cite{cacheenabled} Bastug \emph{et al.} couple the caching problem with the physical layer. In \cite{diststorage} Altman \emph{et al.} compare the expected cost of obtaining the complete data under uncoded and coded data allocation strategies for caching problem. \BS{Cache placement with the help of stochastic geometry and optimizing the allocation of storage capacity among files in order to minimize the cache miss probability problem is presented by Avrachenkov \emph{et al.} in \cite{optimizatioofcaching}. A combined caching scheme where part of the available cache space is reserved for caching the most popular content in every small base station, while the remaining is used for cooperatively caching different partitions of the less popular content in different small base stations, as a means to increase local content diversity is proposed by Chen \emph{et al.} in \cite{cooperativecaching}. In \cite{utility} Dehghan \emph{et al.} associate each content a utility, which is a function of the corresponding content hit probabilities and propose utility-driven caching by formulating an optimization problem to maximize the sum of utilities over all contents.}

The \emph{main contribution} of this paper is to find optimal placement strategies that maximize total hit probability in heterogeneous cellular networks. Our focus is on heterogeneous cellular networks in which an operator wants to \BSwcnc{jointly} optimize the cached content in macro base stations (MBSs) and small base stations (SBSs) with different storage capacities. \BSwcnc{This problem is not convex.} \BSwcncfin{Therefore,} we provide a heuristic scheme and optimize only one type of cache (e.g. SBS) with the information coming from other types of caches (e.g. MBS and other SBSs with different cache capacities) at each iteration step.

We show that whether MBSs use the optimal deployment strategy or store ``the most popular content", has \BSwcnc{no} impact on the total hit probability after deploying the SBSs with optimal content placement policies. We show that it is crucial to optimize the content placement strategy of the SBSs in order to maximize the overall performance. \BSwcnc{We show that} heuristic policies like storing the popular content that is not yet available in the MBSs \BS{result} in significant performance penalties.

In Section \ref{model} we start the paper with model and problem definition. In Section \ref{optimization} we present the optimal \BSr{content} placement strategy problem and give required tools to solve it. We conclude this section by \BSwcncfin{providing an iterative heuristic scheme for the joint optimization problem.} In Section \ref{sec:performance} we continue with performance evaluation of the optimal placement strategies for different probabilistic deployment scenarios. In Section \ref{discussion} we conclude the paper with discussions.

%
%
%
%
\section{Model and Problem Definition}
\label{model}
\BSwcnc{Throughout the paper we will be interested in two types of base stations, namely MBSs and SBSs. However, \BSr{in} this section we will give a more general formulation as it is possible to have SBSs with different storage capacities in some topologies.} We consider a heterogeneous cellular network with L-different types of base stations in the plane. \JG{These base stations are distributed according to a homogeneous spatial process~\cite{baccelli1}. \BSwcnc{All base stations have a coverage radius $r$.} Let $\mathcal{N}_i$, $i=1,\dots, L$, denote the number of base stations of type $i$ that are covering a user at an arbitrary location in the plane. Furthermore, let $p_{n_i}^{(i)} := \mathbb{P}[\mathcal{N}_i = n_i]$ and $p_{\bm{n}} := \mathbb{P}[\bm{\mathcal{N}} = \bm{n}]$, where $\bm{\mathcal{N}}=(\mathcal{N}_1,\dots, \mathcal{N}_L)$ and $\bm{n} = \left(n_1, \dots, n_L\right).$} \JGr{Each base station is equipped with a cache and we will refer to caches in base stations of type $i$ as type$-i$ caches. Caches of type $i$ have capacity $K_i \geq 1$, $i = 1, \dots, L$, meaning that they can store $K_i$ files from the content library $\mathcal{C} = \{c_1, c_2, \dots, c_J\}$. All files have equal size.

Users independently request files from $\mathcal{C}$ according to probability distribution $a_1, \dots, a_J$, where $a_j$ is the probability of requesting file $c_j$. 
The probability that file $c_j$ is requested is denoted as $a_j$.}
Without loss of generality, $a_1 \geq a_2 \geq \dots \geq a_J$.

We use the probabilistic content placement policy of~\cite{optimalgeographic}. \JGr{In this policy each cache is storing a random subset of the files in the content library. In particular, the probability that a type$-i$ cache is storing file $c_j$ is independent and identically distributed over all type$-i$ caches. Placement is also independent over caches of different types.} We denote the probability that the content $c_j$ is stored at a type$-i$ cache as
\begin{equation}
b_j^{(i)} := \mathbb{P}\left(c_j \text{ stored in  type$-i$ cache}\right),
\end{equation}
and the \BSr{content} placement strategy $\bm{b}^{(i)} = (b_1^{(i)}, \dots, b_J^{(i)})$  as a J-tuple for any type$-i$ cache. 

\JGr{One means obtaining a placement according this distribution is as follows.} The memory of a type$-i$ cache is divided into $K_i$ continuous memory intervals of unit length. Then $b_j^{(i)}$ values fill the cache memory sequentially and continue filling the next slot if not enough space is available in the memory slot that it has started filling in as in the end completely covering the $K_i$ memory intervals. Then, for any type$-i$ cache, a random number from the interval $[0,1]$ is picked and the intersecting $K_i$ files are cached. An example is shown in Figure \ref{realization}.
\begin{figure}
\centering
\includegraphics[width=0.7\columnwidth]{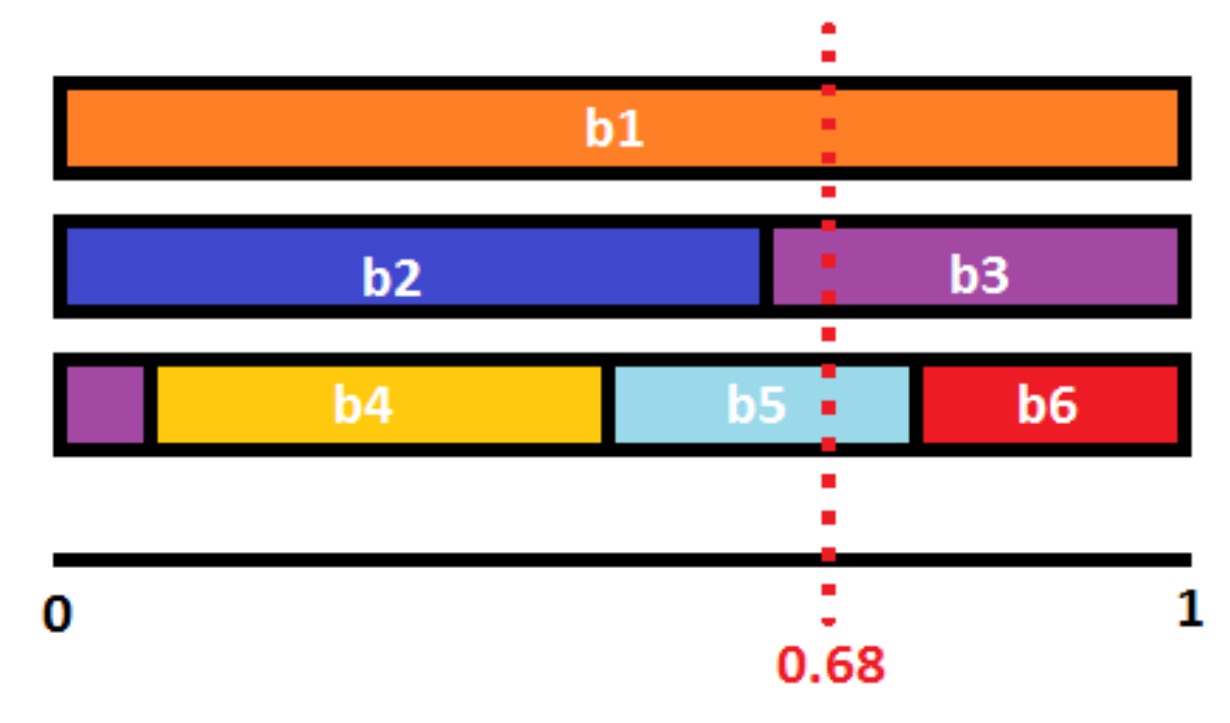}
\caption{A realization of the probabilistic placement policy ($J = 6$ and $K = 3$). A random number is picked (0.68 in this case) and the vertical line intersects with the optimal $\bm{b}^{(i)}$ values. From this figure, we conclude that the content subset $\{c_1, c_3, c_5\}$ will be cached.}
\label{realization}
\end{figure}
\section{Content Placement Problem - \BSwcnc{Individual} Optimization}
\label{optimization}
\subsection{Formulation of the problem}
We start this section with defining our performance metric and the formulation of the optimization problem. The performance metric is the total hit probability which is the probability that a user will find the content that he requires in one of the {\it caches} that he is covered by. We assume that the \BSr{content} placement strategy for the probability distribution over $J$ files through all $L-1$ types is \JG{fixed and} known and we will solve the optimization \JG{problem for} only one type. \JG{In this section, without loss of generality, we consider the optimization of type$-1$.}
For notational convenience, superscript $^c$ \JG{in the notation $b_j^{(i)^c}$} indicates that the \BSr{content} placement strategy \JG{for type$-i$} is known and constant. Then, the total hit probability is given by
\begin{equation}
f\left(\bm{b}^{(1)}\right) = 1 - \sum_{j=1}^J a_j \sum_{n_1= 0}^{\infty}(1-b_j^{(1)})^{n_1}p_{n_1}^{(1)} q(j, n_1),
\label{hitprob}
\end{equation}
where
\begin{align}
&q(j, n_1) = P(\text{non type$-1$ caches miss the file $c_j$}) \nonumber\\
&= \sum_{n_2=0}^{\infty} \dots \sum_{n_L=0}^{\infty} \prod_{l = 2}^L (1-b_j^{(l)^c})^{\BS{n_l}} P(\bm{\mathcal{N}}_2^L = \bm{n}_2^L \vert N_1 = n_1),\nonumber
\end{align} 
where $\bm{n}_2^L = \left(n_2, \dots, n_L\right)$ and $\bm{\mathcal{N}}_2^L = (\mathcal{N}_2, \dots, \mathcal{N}_L)$ are $(L-1)$-tuples representing the coverage vectors of non type$-1$ caches.

We define the optimization problem to find the optimal \BSr{content} placement strategy maximizing the total hit probability for a type$-1$ cache as follows:
\begin{prb}
\label{modprb}
\begin{align}
&\max \text{ } f\left(\bm{b}^{(1)}\right)\nonumber\\
&\text{ }\mathbf{s.t.}\quad b_1^{(1)} + \dots + b_J^{(1)} = K_1, \quad b_j^{(1)} \in [0,1],\quad \forall j. \label{constraints}
\end{align}
\end{prb}
\subsection{Solution of the optimization problem}
In this section, we will analyze the structure of the optimization problem.
\begin{lem}
\label{convex}
Problem \ref{modprb} is a convex optimization problem.
\end{lem}
\begin{proof}
The objective function is separable with respect to (w.r.t.) $b_1^{(1)}, \dots, b_J^{(1)}$ (Replace $\sum_{j=1}^J a_j = 1$ and rewrite the objective function \eqref{hitprob}.). $1-f\left(\bm{b}^{(1)}\right)$ is increasing and convex in $b_j^{(1)}$, $\forall j$. Hence, it is a convex function of $\left(b_1^{(1)}, \dots, b_J^{(1)}\right)$.
\end{proof}
We already showed that $1-f\left(\bm{b}^{(1)}\right)$ is convex by Lemma \ref{convex} and the constraint set is linear as given in \eqref{constraints}. Thus, the Karush-Kuhn-Tucker (KKT) conditions provide necessary and sufficient conditions for optimality. The Lagrangian function corresponding to Problem \ref{modprb} becomes
\begin{align}
&L\left(\bm{b}^{(1)}, \nu, \bm{\lambda}, \bm{\omega}\right) = \sum_{j=1}^J a_j \sum_{n_1= 0}^{\infty}(1-b_j^{(1)})^{n_1} p_{n_1}^{(1)} q(j, n_1) \nonumber\\ 
&+ \nu \left(\sum_{j=1}^J b_j^{(1)} - K_1\right) - \sum_{j=1}^J \lambda_j b_j^{(1)}  + \sum_{j=1}^J \omega_j \left(b_j^{(1)} - 1\right), \nonumber\\
\end{align}
where $\bm{b}^{(1)}$, $\bm{\lambda}$, $\bm{\omega} \in \mathbb{R}_+^J$ and $\nu \in \mathbb{R}$.

Let $\bar{\bm{b}}^{(1)}$, $\bar{\bm{\lambda}}$, $\bar{\bm{\omega}}$ and $\bar{\nu}$ be primal and dual optimal. The KKT conditions for Problem \ref{modprb} state that
\begin{align}
\sum_{j=1}^J \bar{b}^{(1)}_j &= K_1, \label{kkt2}\\
0 \leq \bar{b}^{(1)}_j &\leq 1, \quad \forall j = 1,\dots, J, \label{kkt1}\\
\bar{\lambda}_j &\geq 0, \quad \forall j = 1,\dots, J,\label{kkt3}\\
\bar{\omega}_j &\geq 0,\quad \forall j = 1,\dots, J,\label{kkt4}\\
\bar{\lambda}_j \bar{b}^{(1)}_j &= 0,\quad \forall j = 1,\dots, J,\label{kkt5}\\
\bar{\omega}_j \left(\bar{b}^{(1)}_j - 1\right) &= 0, \quad \forall j = 1,\dots, J,\label{kkt6}
\end{align}
\begin{align}
-a_j \sum_{n_1= 0}^{\infty}n_1(1-b_j^{(1)})^{n_1-1}p_{n_1}^{(1)} q(j, n_1) + \bar{\nu} - \bar{\lambda}_j + \bar{\omega}_j &= 0, \label{kkt7}\\
\forall j = 1,\dots, &J \nonumber.
\end{align}
\begin{thm}
\label{optsol}
The optimal \BSr{content} placement strategy for Problem \ref{modprb} is
\begin{align}
\label{bfunc}
\bar{b}^{(1)}_j = \left\{
\begin{array}{rl}
1, & \text{if } \bar{\nu} < a_j p_{1}^{(1)} q(j, 1)\\
0, & \text{if } \bar{\nu} > a_j \sum_{n_1=0}^{\infty}n_1 p_{n_1}^{(1)} q(j, n_1),\\
\phi(\bar{\nu}), & \text{otherwise},
\end{array} \right.
\end{align}
where $\phi(\bar{\nu})$ is the solution over $b^{(1)}_j$ of
\begin{equation}
a_j \sum_{n_1= 0}^{\infty}n_1(1-b_j^{(1)})^{n_1-1} p_{n_1}^{(1)} q(j, n_1) = \bar{\nu}, \label{eqsb}
\end{equation}
and $\bar{\nu}$ can be obtained as the unique solution to the additional constraint
\begin{equation}
\bar{b}_1^{(1)} + \dots + \bar{b}_J^{(1)} = K_1. \label{eqK}
\end{equation}
\end{thm}
\begin{proof}
The proof is provided in Appendix \ref{optstr}.
\end{proof}

\BSr{Finding the optimal \BSr{content} placement strategy for all types of caches jointly is an interesting problem. However, this optimization problem is \emph{not convex}. Hence, we provide a heuristic algorithm to see how the overall system performance is improved. The procedure is as follows. At each iteration step we find the optimal strategy for a specific type of cache assuming that the placement strategies for other types of caches are known and fixed. Then we continue with the same procedure for the next type and we continue iterating over different types.} \BSr{In the next section we will evaluate the performance of this scheme.}
\section{Performance Evaluation}
\label{sec:performance}
In this section we will specify different network coverage models and show the performances of the placement strategies. In both models we use Zipf distribution for the file popularities. The probability that a user will ask for content $c_j$ is equal to
\begin{equation}
a_j = \frac{j^{-\gamma}}{\sum_{j=1}^J j^{-\gamma}}, \label{zipfpars}
\end{equation}
where $\gamma > 0$ is the Zipf parameter.
\subsection{Poisson Point Process (PPP) Model}
The caches follow a two-dimensional (2D) spatial homogeneous Poisson process with type$-i$ caches independently distributed in the plane with density $\lambda_i > 0$\JG{, $i = 1, \dots, L$. The} number of type$-i$ caches within radius $r$ follows a Poisson distribution with parameter $t_i = \lambda_i \pi r^2$\JG{, i.e.,}
\begin{align}
p_{n_i}^{(i)} &= P(\text{$n_i$ type$-i$ caches within radius $r$})\nonumber\\
&= \frac{t_i^{n_i}}{n_i!}e^{-t_i}. \label{poisdist}
\end{align}

The user is covered by $n_i$ type$-i$ caches and distribution of the different type of caches is independent of each other. Therefore, the total coverage distribution probability mass function $p_{\bm{n}}$ will be the product of individual probability distributions
\begin{equation}
p_{\bm{n}} = p_{n_1}^{(1)}p_{n_2}^{(2)}\dots p_{n_L}^{(L)}.
\end{equation}

Consider the case of two types of caches in the plane. Type$-1$ caches represent MBSs and type$-2$ caches represent SBSs with $K_1$ and $K_2$-slot cache memories, \BSr{respectively}. The content library size is $J = 100$. We set $K_1 = 1$ and $K_2 = 2$. We set the Zipf parameter $\gamma = 1$ and \JG{taking} $a_j$ \JG{according to} \eqref{zipfpars}. Also we set $\lambda_1 = 0.5$ \JG{ and $r=1$. We will consider various values for the deployment densities of SBSs, namely $\lambda_2$. Our aim is to compare the optimal \BSr{content} placement strategy with various heuristics. Therefore, we will compare various scenarios in the remainder of this subsection.}

For the [OPT/OPT] scenario first we find the optimal solution for MBSs assuming that there are no SBSs in the plane. Solving this optimization problem gives $\bm{\bar{b}}^{(1)} = \left(0.7136, 0.2723, 0.0141, 0, \dots, 0\right)$. The optimal \BSr{content} placement strategy results in storing the first three most popular files into MBSs with given probabilities and the resulting hit probability is $ f\left(\bm{\bar{b}}^{(1)}\right) = 0.1649$. Then we set $\bm{b}^{(1)^c} = \bm{\bar{b}}^{(1)}$, take it as an input, add SBSs into the plane and find the optimal \BSr{content} placement strategy $\bm{\bar{b}}^{(2)}$ for SBSs. \BSr{We observe that applying the iterative optimization procedure, namely repeatedly updating $\bm{\bar{b}}^{(1)}$ and $\bm{\bar{b}}^{(2)}$, does not improve the hit probability.}
\begin{figure}
\centering
\includegraphics[width=0.95\columnwidth]{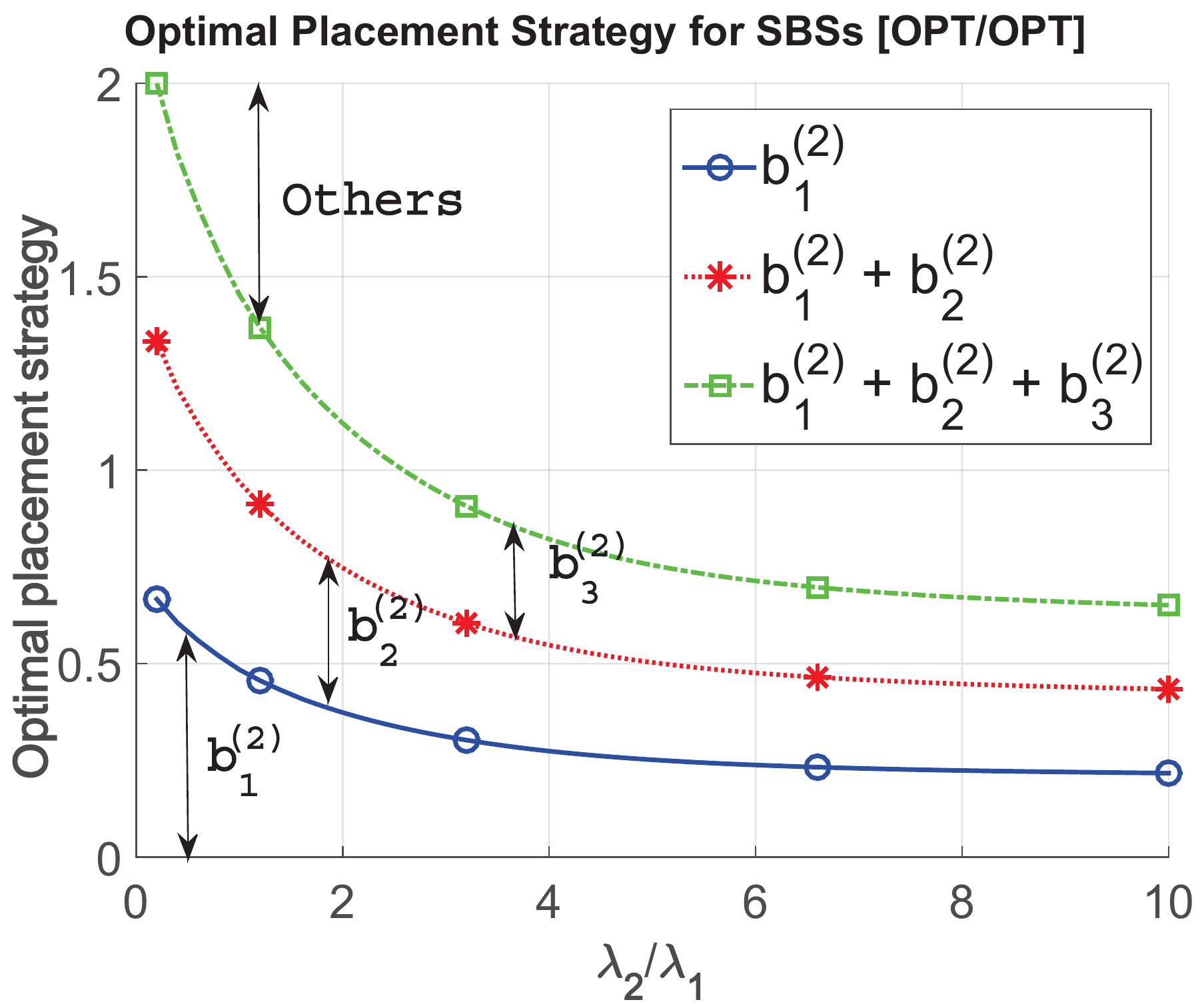}
\caption{Optimal \BSr{content} placement strategy $\bm{\bar{b}}^{(2)}$ for SBSs with different deployment densities $\lambda_2$ [OPT/OPT] (PPP).}
\label{optscsmarttype2zoomed}
\end{figure}
From Figure \ref{optscsmarttype2zoomed}, after deploying SBSs we see that when the deployment density of SBSs is low compared to MBSs, only the first three most popular files are stored in SBSs. As the deployment density of SBSs increases, we see that probability of storing less popular files in the caches increases. 

For the [POP/OPT] scenario we consider a heuristic for the placement in MBSs: The most popular content $c_1$ is stored at MBSs with probability (w.p.) $1$, i.e., $\bm{\bar{b}}^{(1)} = \left(1, 0, \dots, 0\right)$. The resulting hit probability is $ f\left(\bm{\bar{b}}^{(1)}\right) = 0.1527$. It is important to note that the hit probability under this heuristic policy is not significantly different from the optimal hit probability.

Next we add SBSs. We set $\bm{b}^{(1)^c} = \bm{\bar{b}}^{(1)}$, take it as an input, add SBSs into the plane and find the optimal \BSr{content} placement strategy $\bm{\bar{b}}^{(2)}$ for SBSs.
\begin{figure}
\centering
\includegraphics[width=0.95\columnwidth]{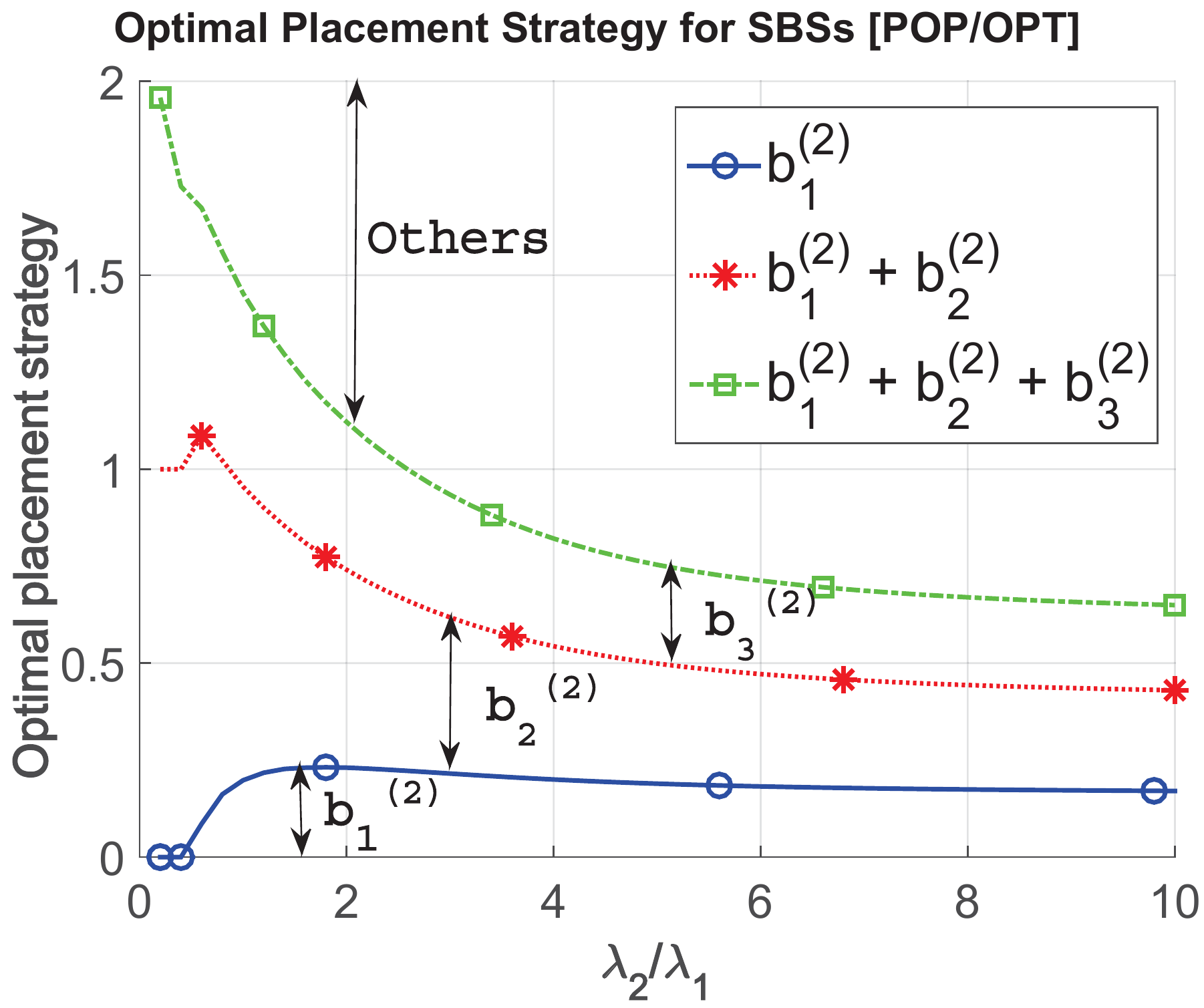}
\caption{Optimal \BSr{content} placement strategy $\bm{\bar{b}}^{(2)}$ for SBSs with different deployment densities $\lambda_2$ [POP/OPT] (PPP).}
\label{optschugelibtype2mostpop}
\end{figure}
\begin{figure}
\centering
\includegraphics[width=0.95\columnwidth]{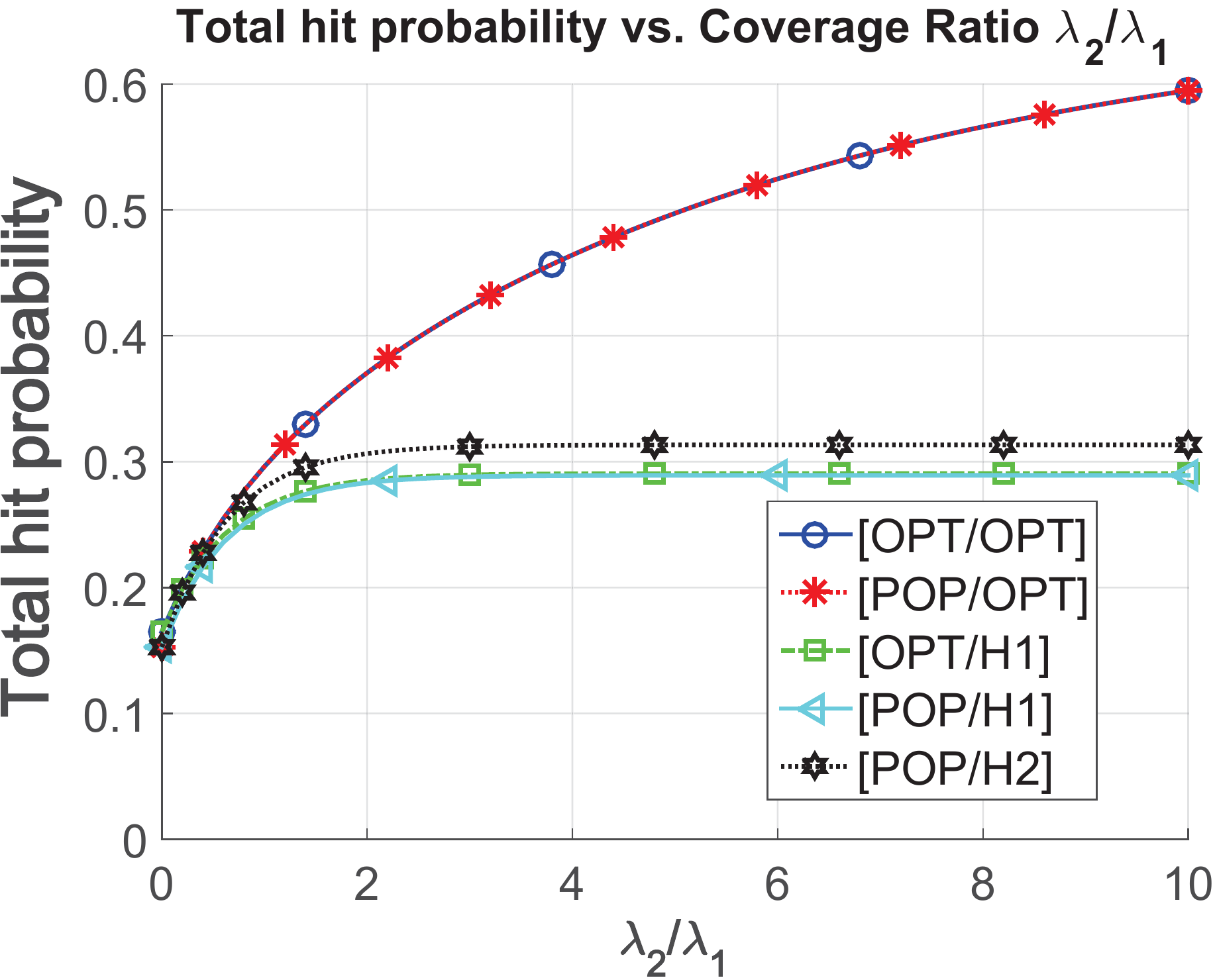}
\caption{The total hit probability evaluation for different SBS deployment densities (PPP).}
\label{hitprobpois}
\end{figure}
From Figure \ref{optschugelibtype2mostpop}, after deploying SBSs we see that when the deployment density of SBSs is low compared to MBSs, only $c_2$ and $c_3$ are stored in SBSs with probability $1$ in order to increase the chance of the user to get these files. As the deployment density of SBSs increases, we see that probability of storing $c_1$ in SBSs slightly increases. However, we see that since we cannot get $c_2$ and $c_3$ from any of the MBSs, probabilities of deploying $c_2$ and $c_3$ are always higher than $c_1$. Probability of storing other files also increases as the $\lambda_2$ increases.

From Figure \ref{hitprobpois}, \BSr{we see that the difference between the hit probabilities using the [OPT/OPT] and the [POP/OPT] heuristics decreases as the density of SBSs increases, i.e. we can use a heuristic placement policy for MBSs as finding the optimal \BSr{content} placement strategy for SBSs compensate the performance penalty caused by ill-adjusted content placement of MBSs}. However, as we will see next, no straightforward heuristic seems to exist for the \BSr{content} placement policy in SBSs.

The first heuristic is to use is to store the two most popular files in SBSs, denoted by [OPT/H1] and [POP/H1] for an optimal and a heuristic policy in MBS, respectively. The second heuristic is to store the second and third most popular files in SBSs [POP/H2]. From Figure~\ref{hitprobpois} it is clear that these heuristic policies achieve significantly lower hit probability than the optimal policy.

\subsection{M-or-none deployment model}
Type$-1$ caches represent MBSs and follow a two-dimensional (2D) spatial homogeneous Poisson process with density $\lambda_1 > 0$, and number of MBSs within radius $r$ follows a Poisson distribution satisfying \eqref{poisdist} with $i = 1$.

We assume that if a user is covered by at least one MBS (type$-1$ cache), then it should be covered by $M$ caches in total. As a result, network operators serve users with providing them $M$ caches as long as they are connected to one of the macro base stations. If a user is not in the coverage area of a MBS, then it doesn't receive any service from other caches either at all. Therefore, we have
\begin{equation*}
P\left(\bm{\mathcal{N}}_2^L = \bm{n}_2^L \vert N_1 = 0\right) = \left\{
\begin{array}{rl}
0 & \text{if } \max\{\bm{n}_2^L\} > 0,\\
1 & \text{if } \max\{\bm{n}_2^L\} = 0,
\end{array} \right.
\end{equation*}
and 
\begin{equation*}
P\left(\bm{\mathcal{N}}_2^L = \bm{n}_2^L \vert N_1 = n_1 \right) = \left\{
\begin{array}{rl}
0 & \text{if } \sum_{l=2}^L n_l \neq (M - n_1)^+,\\
1 & \text{if } \sum_{l=2}^L n_l = (M - n_1)^+,
\end{array} \right.
\end{equation*}
when $n_1 > 0$.

Consider the case of two types of caches in the plane. Type$-1$ caches represent MBSs and type$-2$ caches represent SBSs with $K_1$ and $K_2$-slot cache memories, \BSr{respectively}. The content library size is $J = 100$. We set $K_1 = 1$ and $K_2 = 2$. Also, we set the Zipf parameter $\gamma = 1$ and taking $a_j$ according to \eqref{zipfpars}. Also we set $\lambda_1 = 0.5$ and $r = 1$.

For the [OPT/OPT] scenario, optimal \BSr{content} placement strategy for MBSs is $\bm{\bar{b}}^{(1)} = \left(0.7136, 0.2723, 0.0141, 0, \dots, 0\right)$ and the resulting hit probability is $ f\left(\bm{\bar{b}}^{(1)}\right) = 0.1649$. Then we set $\bm{b}^{(1)^c} = \bm{\bar{b}}^{(1)}$, take it as an input and find the optimal \BSr{content} placement strategy $\bm{\bar{b}}^{(2)}$ for different $M$ values. \BSr{We observe that using the iterative optimization procedure, namely repeatedly updating $\bm{\bar{b}}^{(1)}$ and $\bm{\bar{b}}^{(2)}$, does not improve the hit probability.}
\begin{figure}
\centering
\includegraphics[width=0.95\columnwidth]{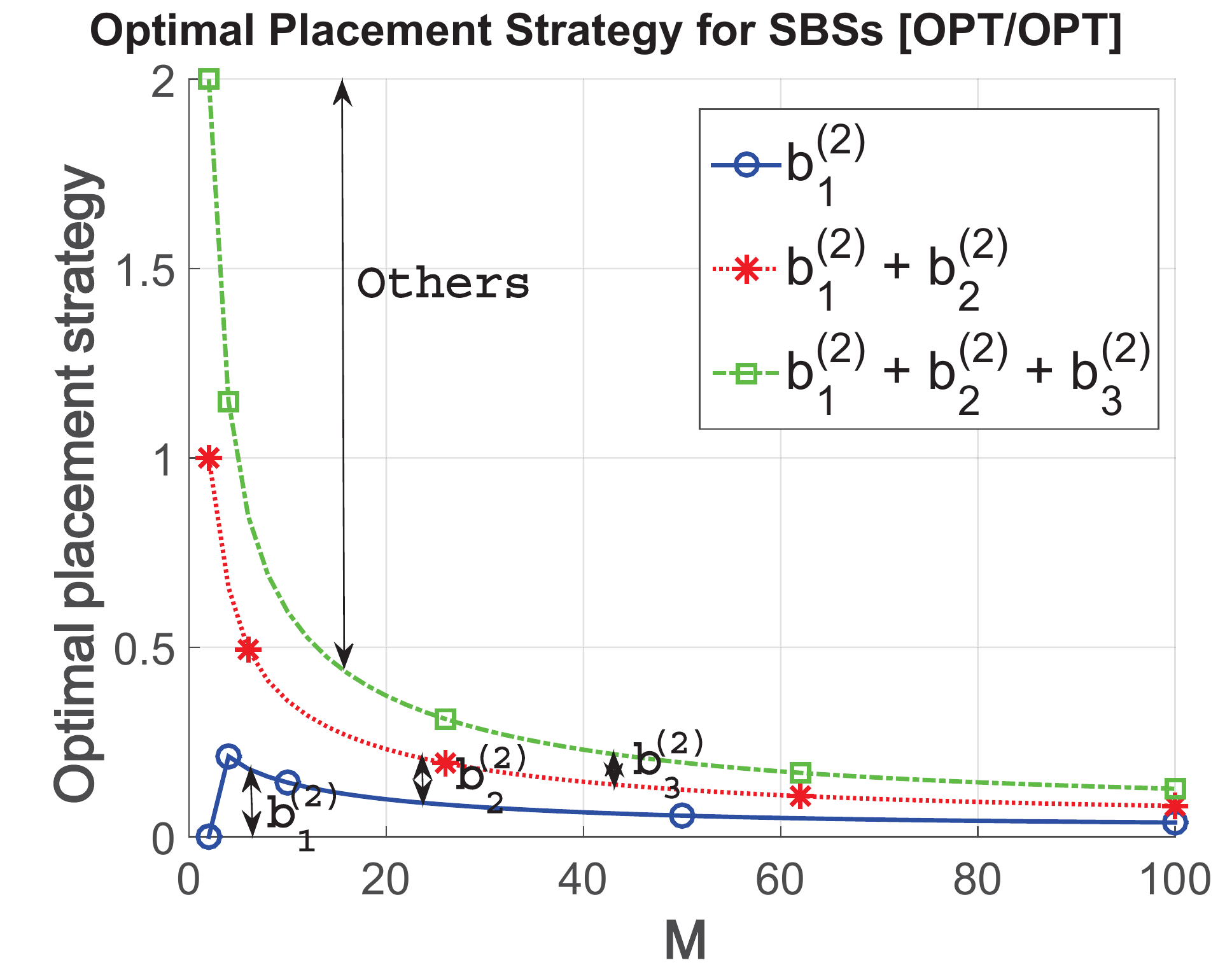}
\caption{Optimal \BSr{content} placement strategy $\bm{\bar{b}}^{(2)}$ for SBSs for different $M$ values [OPT/OPT] (M-or-None).}
\label{cachingpolicyMorNoneSmart}
\end{figure}
From Figure \ref{cachingpolicyMorNoneSmart}, we see that the probability of storing less popular files increases as $M$ increases.

For the [POP/OPT] scenario, we have $\bm{\bar{b}}^{(1)} = \left(1, 0, \dots, 0\right)$ and the resulting hit probability is $ f\left(\bm{\bar{b}}^{(1)}\right) = 0.1527$. Then we set $\bm{b}^{(1)^c} = \bm{\bar{b}}^{(1)}$, take it as an input, add SBSs into the plane and find the optimal \BSr{content} placement strategy $\bm{\bar{b}}^{(1)}$ for SBSs.
\begin{figure}
\centering
\includegraphics[width=0.95\columnwidth]{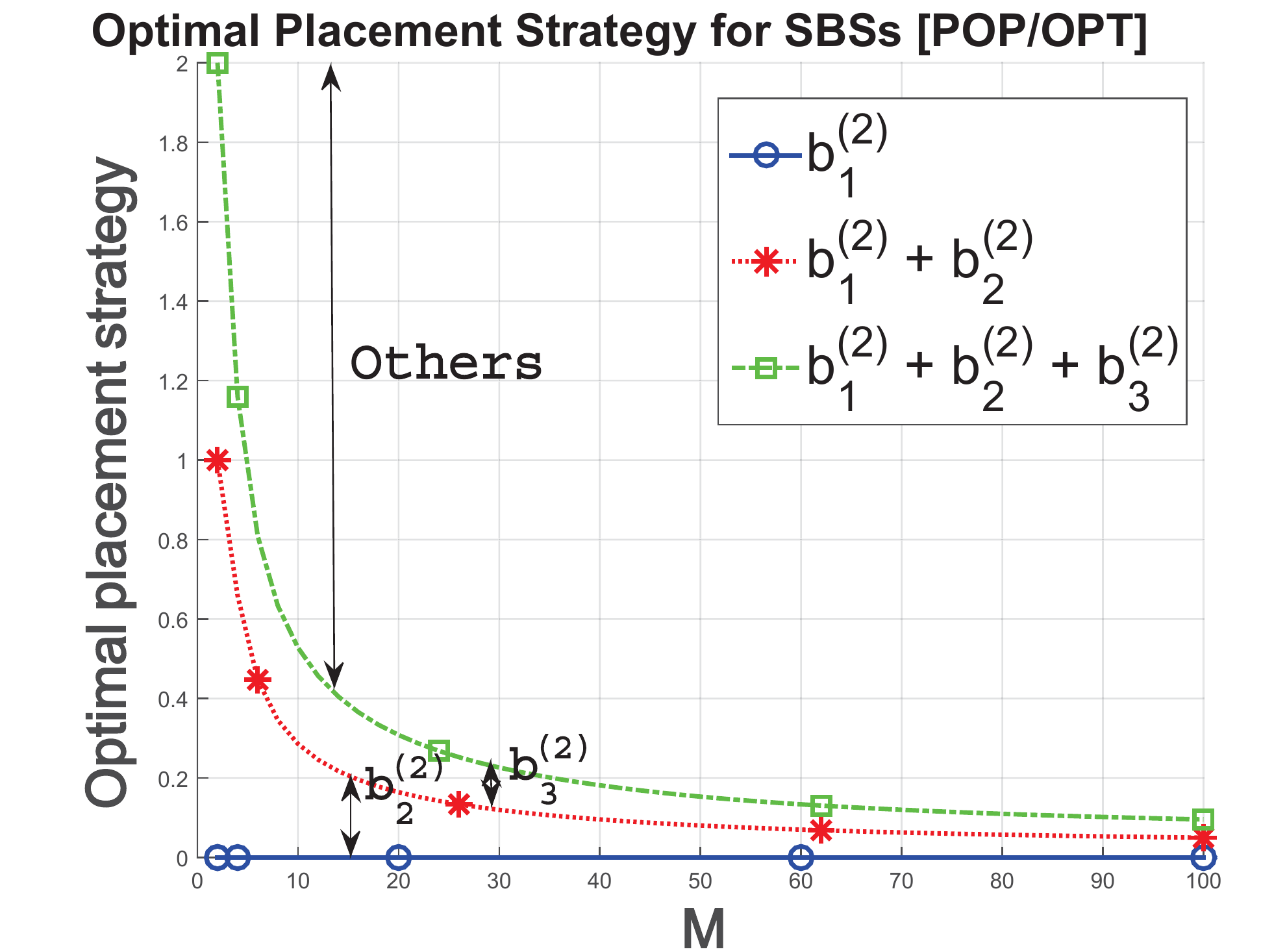}
\caption{Optimal \BSr{content} placement strategy $\bm{\bar{b}}^{(2)}$ for SBSs for different $M$ values [POP/OPT] (M-or-None).}
\label{cachingpolicyMorNoneMostpop}
\end{figure}
\begin{figure}
\centering
\includegraphics[width=0.95\columnwidth]{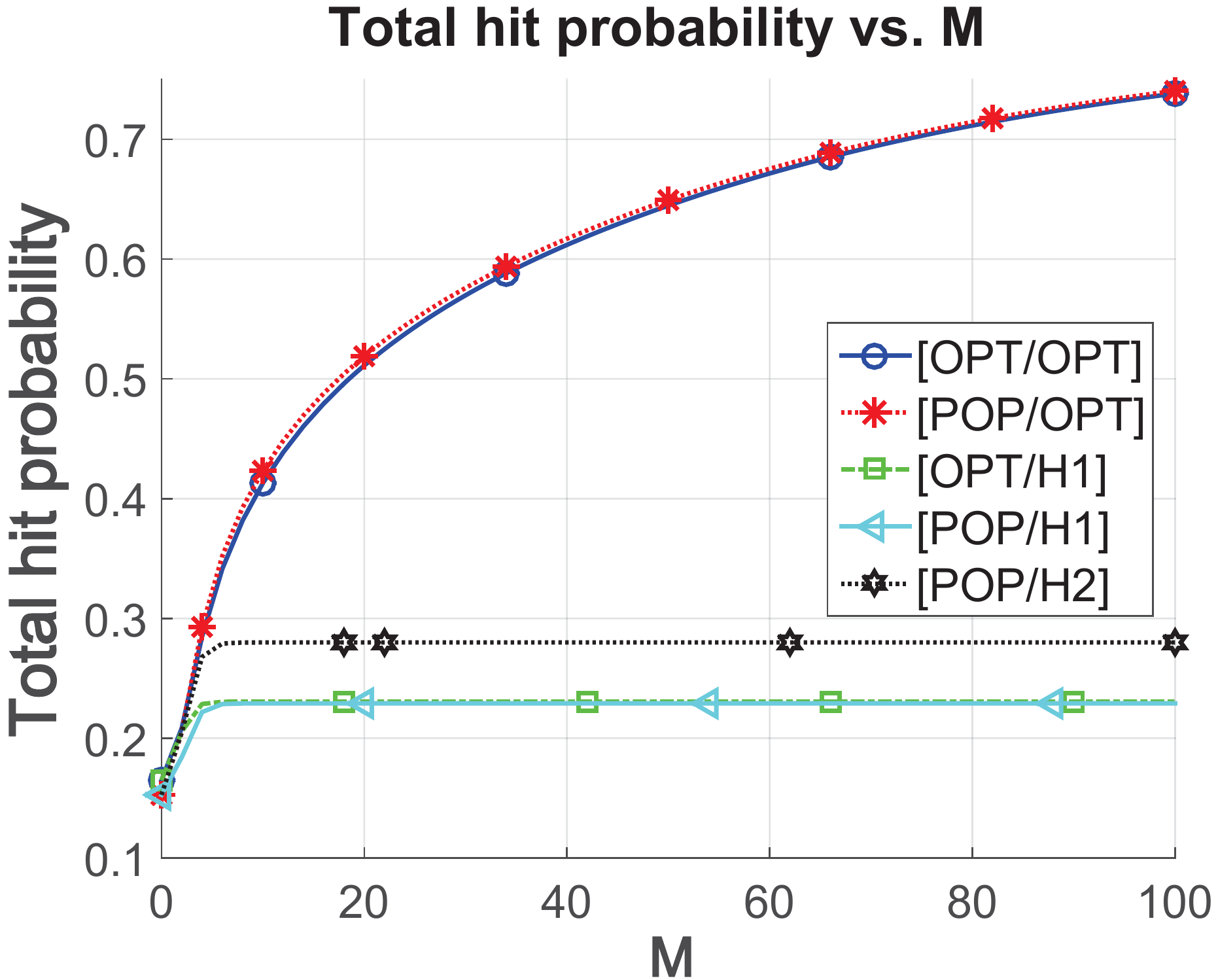}
\caption{The total hit probability evaluation for different $M$ values (M-or-None).}
\label{hitprobMorNone}
\end{figure}
From Figure \ref{cachingpolicyMorNoneMostpop}, as $c_1$ is stored in MBSs w.p. 1 and SBSs are present in the system only when $n_1 > 0$, $c_1$ is never stored at SBSs. We see that probability of storing $c_2$ and $c_3$ decreases and probability of storing other files increases as $M$ increases.

From Figure \ref{hitprobMorNone}, we see that the hit probabilities under the [OPT/OPT] and the [POP/OPT] are identical and increase as $M$ increases, i.e. we can use a heuristic placement policy for MBSs. For heuristic SBS deployment policies, [OPT/H1], [POP/H1] and [POP/H2] policies achieve significantly lower hit probability than the optimal policy.

\section{Discussion and Conclusion}
\label{discussion}
In this paper \BSwcnc{we have shown} that whether MBSs use the optimal deployment strategy or store ``the most popular content", has very limited impact on the total hit probability \BSwcnc{if the SBSs are using the optimal deployment strategy}. It is important to optimize the content placement strategy of the SBSs and the total hit probability is increased significantly when the SBSs use the optimal deployment strategy. It is shown that heuristic policies for SBSs like storing the popular content that is not yet available in the MBS results in significant performance penalties. To conclude, using the optimal deployment strategy for the SBSs is crucial and it ensures the overall network to have the greatest possible total hit probability independent of the deployment policy of MBSs. \BSwcncfin{Finally, we have shown that solving the individual problem to find optimal \BSr{content} placement strategy for different types of base stations iteratively, namely repeatedly updating the placement strategies of the different types, does not improve the hit probability.}

\appendices
\section{Proof of Theorem~\ref{optsol}} \label{optstr}
From \eqref{kkt5}, \eqref{kkt6} and \eqref{kkt7}, we have
\begin{equation}
\bar{\omega}_j = \bar{b}_j^{(1)} \left[a_j \sum_{n_1=0}^{\infty} n_1 (1-b_j^{(1)})^{n_1 -1} p_{n_1}^{(1)} q(j, n_1) - \bar{\nu}\right], \label{omegaeq}
\end{equation}
which, when insterted into \eqref{kkt6}, gives
\begin{equation}
\bar{b}_j^{(1)}\left(\bar{b}_j^{(1)} - 1\right)\left[a_j \sum_{n_1=0}^{\infty}n_1 (1-b_j^{(1)})^{n_1 -1} p_{n_1}^{(1)} q(j,n_1)\right] = \bar{\nu} \label{star}.
\end{equation}
From \eqref{star}, we see that $0 < \bar{b}_j^{(1)} < 1$ only if 
\begin{equation*}
\bar{\nu} = a_j \sum_{n_1=0}^{\infty}n_1 (1-b_j^{(1)})^{n_1 -1} p_{n_1}^{(1)} q(j, n_1).
\end{equation*}
Since we know that $0 \leq b_j^{(i)} \leq 1$, this implies that
\begin{align*}
\bar{\nu} \in \left[a_j p_{1}^{(1)} q(j, 1),\text{ } a_j \sum_{n_1=0}^{\infty} n_1 p_{n_1}^{(1)} q(j,n_1) \right].
\end{align*}
If $\bar{\nu} < a_j p_{1}^{(1)}q(j, 1)$, we have
\begin{equation*}
\bar{\omega}_j = \bar{\lambda}_j + a_j \sum_{n_1=0}^{\infty} n_1 (1-b_j^{(1)})^{n_1 -1} p_{n_1}^{(1)} q(j, n_1) - \bar{\nu} > 0.
\end{equation*}
Thus, from \eqref{kkt6}, we have $\bar{b}_j^{(1)} = 1$. Similarly, if $\bar{\nu} > a_j \sum_{n_1=0}^{\infty} n_1 p_{n_1}^{(1)} q(j, n_1)$, we have
\begin{equation*}
\bar{\lambda}_j = \bar{\omega}_j +  \bar{\nu} -a_j \sum_{n_1=0}^{\infty}n_1 (1-b_j^{(1)})^{n_1 -1} p_{n_1}^{(1)} q(j, n_1) > 0.
\end{equation*}
Hence, from \eqref{kkt5}, we have $\bar{b}_j^{(1)} = 0$.

Finally, since $\sum_{j=1}^J {b}_j^{(1)}$ is a decreasing function in $\nu$, solving $J$ equations of \eqref{eqsb} satisfying \eqref{eqK} give the unique solution $\bar{\nu}$.

\end{document}